\documentclass[reqno]{amsart}

\usepackage[hyphens]{url}
\usepackage[style=numeric,sorting=none,backend=biber,maxbibnames=99]{biblatex}
\usepackage{tiz-note}
\bibliography{cond}
\usepackage{tiz-paper}

\usepackage[links,references]{latex/agda}
\usepackage{newunicodechar}
\newunicodechar{⊎}{\ensuremath{\mathrel{\uplus}}}
\newunicodechar{ℝ}{\ensuremath{\mathnormal{\mathbb S}}}
\newunicodechar{ℕ}{\ensuremath{\mathnormal{\mathbb N}}}
\newunicodechar{ℤ}{\ensuremath{\mathnormal{\mathbb Z}}}
\newunicodechar{¹}{\ensuremath{\mathnormal{^1}}}
\newunicodechar{₀}{\ensuremath{\mathnormal{_0}}}
\newunicodechar{⊤}{\ensuremath{\mathnormal{\top}}}
\newunicodechar{μ}{\ensuremath{\mathnormal{\eta}}}
\newunicodechar{∞}{\ensuremath{\mathnormal{\infty}}}

\newcommand{\lcon}{\fs{\cons{left}}}
\newcommand{\rcon}{\fs{\cons{right}}}
\newcommand{\pcon}{\fs{\func{path}}}
\newcommand{\Path}{\fs{\Dat{Path}}}
\newcommand{\pconp}[1]{\pcon~(#1)}
\newcommand{\pappSym}{\elim@}
\newcommand{\papp}[2]{#1.\pappSym~#2}

\newcommand{\Sph}[1]{\Dat{\mathbb S^#1}}
\newcommand{\Tori}{\Dat{Torus}}
\newcommand{\coe}{\fs{\elim{coe}}}
\newcommand{\squeeze}[2]{{#1 \wedge #2}}
\newcommand{\kwwhere}{\kw{where}}
\newcommand{\kwdata}{\kw{data}}
\newcommand{\kwwith}{\kw{with}}
\newcommand{\kwcase}{\kw{case}}
\newcommand{\poscon}{\cons{pos}}
\newcommand{\negcon}{\cons{neg}}
\newcommand{\zerocon}{\cons{zero}}
\newcommand{\succon}{\cons{suc}}
\newcommand{\basecon}{\cons{base}}
\newcommand{\loopcon}{\cons{loop}}
\newcommand{\pointcon}{\cons{point}}
\newcommand{\linecon}[1]{\cons{line#1}}
\newcommand{\facecon}{\cons{face}}
\newcommand{\Nat}{\Dat{\mathbb N}}
\newcommand{\NatInf}{\Dat{\mathbb N^\infty}}

\newcommand{\Int}{\Dat{\mathbb Z}}
\newcommand{\II}{\Dat{\mathbb I}}

\newcommand{\vars}[1]{\textsf{vars}(#1)}
\newcommand{\replace}[2]{\textsf{rename}(#1,#2)}
\newcommand{\confluence}[5]{\textsf{confluent}(#1,#2,#3,#4,#5)}

\newcommand{\CubicalAgda}{{\texttt{Cubical} \texttt{Agda}}}

\usepackage[hyphenbreaks]{breakurl}
\DeclareRobustCommand{\SkipTocEntry}[5]{}

\begin{document}

\title{(Co)condition hits the Path}

\author{Tesla Zhang}
\address[The paper is written while the author is at The Pennsylvania State University, but the email address under that institution no longer works.]{Carnegie Mellon University}
\email{teslaz@cmu.edu}

\author{Valery Isaev}
\address{JetBrains Research}
\email{valery.isaev@jetbrains.com}

\begin{abstract}
We propose an enhancement to inductive types and records in a dependent type theory, 
namely (co)conditions. With a primitive interval type,
conditions generalize the cubical syntax of higher
inductive types in homotopy type theory, while coconditions generalize the
cubical path type. (Co)conditions are also useful without an interval type.
The duality between conditions and coconditions is presented in an interesting way:
The elimination principles of inductive types with conditions can be internalized with
records with coconditions and vice versa.

However, we do not develop the metatheory of conditions and coconditions in this paper.
Instead, we only present the type checking.
\keywords{
inductive types,
homotopy type theory,
cubical type theory,
conditions,
confluence
}
\end{abstract}

\maketitle
\tableofcontents
\section{Introduction}%
\label{sec:intro}

Inductive types are useful abstractions of data in functional programming languages, where data are placed into a labeled sum of products.
Each label along with the product it carries is called a \textit{constructor}.
The constructors define how instances of the type could be constructed.
Many functional programming languages such as Agda, OCaml, Standard ML\footnote{Haskell is not listed
as its ``data types'' are coinductive due to laziness.}
have adopted inductive types as their fundamental way of defining data types.

\subsection{Quotients with definitional equality} 
\label{sub:intro-int}

One property of inductive types is that constructors are \textit{disjoint} from each other.
While it may be useful in some cases, the disjointedness can also cause problems. For example,
we may define the type of integers with two constructors (one for positive integers and one for negative integers), 
both taking a natural number as its argument.

We would like to define the integer type with a constructor mapping natural number $n$ to $+n$
and a constructor mapping natural number $n$ to $-n$, but this definition will result in two distinct
zeros (one positive and one negative).
To eliminate one of the zeros, we may change the meaning of one constructor by shifting the
corresponding integer value by one to get rid of one of the zeros.
In the following definition (adapted from the built-in natural number definition in Agda~\cite{Agda}),
the mathematical interpretation of $\cons{negsuc}~n$ is $-(n+1)$:

\begin{figure}[h]
\centering
\subfloat{
$\begin{aligned}
&\kwdata~\Nat:\UU \\[-0.3em]
& \mid~\zerocon \\[-0.3em]
& \mid~\succon~\Nat
\end{aligned}$}
\qquad
\subfloat{
$\begin{aligned}
&\kwdata~\Dat{\mathbb Z'}:\UU \\[-0.3em]
& \mid\poscon~\Nat \\[-0.3em]
& \mid\cons{negsuc}~\Nat
\end{aligned}$}
\end{figure}

We have to use this workaround because we cannot just identify two instances of a certain type.
The disjointedness of constructors makes such identification impossible.
The situation of integer can be generalized to any quotients --- the general idea of identifying
distinct values of a certain set to form a new set.
Quotients are very useful in mathematics, but they are not possible with inductive types.

We present \textit{conditions}, an extension to inductive types,
that allows identification of distinct values under some restrictions.
A condition is a \textit{partial} function defined by pattern matching \textit{attached} to a constructor,
whose type must match the constructor.
The constructor will reduce according to the partial function and will be canonical if stuck.
With conditions, we can define integers like this\footnote{
The $\kwwith$ keyword here is unrelated to the ``with construct'' in Agda~\cite[\S 2.3]{AgdaOri}.}:

\lessSpace{-1.2}
\begin{align*}
&\kwdata~\Int:\UU \\[-0.3em]
& \mid\poscon~\Nat \\[-0.3em]
& \mid\negcon~\Nat~
    \kwwith~\zerocon ⇒ \poscon~\zerocon
\end{align*}

The condition here is $\kwwith~\zerocon ⇒ \poscon~\zerocon$, attached to the constructor $\negcon$.
Observe the type --- the condition maps $\Nat$ to $\Int$, which fits the type of $\negcon$.
Due to the condition, $\negcon~\zerocon$ will reduce to $\poscon~\zerocon$.
So, $\negcon~(\succon~\zerocon)$ is a canonical instance of $\Int$, while $\negcon~\zerocon$
is an instance of $\Int$ that evaluates to $\poscon~\zerocon$ --- another canonical instance of $\Int$.

This definitional equality on integers makes it closer to its definition in mathematics.
In general, any identification (between distinct instances of an inductive type)
that can be represented as a pattern matching can also be defined with conditions.


\subsection{Higher inductive types} 
\label{sub:intro-hits}
In \textit{homotopy type theory}~\cite{hottbook}, types are interpreted as
homotopy types of spaces with their instances interpreted as points in the space.
The identity type in Martin-L\"of's type theory~\cite{MLTT,MLTT2}
(which homotopy type theory is based on) is interpreted as the \textit{path space},
with all relevant rules interpreted homotopically.
The inductive types one can define in those type theories are all discrete spaces because the
instances are distinct, and the identity proofs on these instances are trivial.

To define nontrivial paths in types, inductive types are generalized to \textit{higher inductive types}
(HITs), where constructors define not only the construction of the \textit{instances} of the type,
but also the \textit{identity relation} of the instances.
Simple constructors are called \textit{point constructors}, and they describe how points of the space are generated.
Constructors defining identity relations are called \textit{path constructors}, and they describe how these points 
are connected in the space.
Functions in homotopy type theory are all continuous --- the path constructors must be \textit{preserved}.

The idea of HITs looks similar to conditions.
For example, the integer type can be defined as a HIT, with two point-constructors $\poscon$, $\negcon$
and a path constructor of type $\poscon~\zerocon = \negcon~\zerocon$.
In fact, this is an example in~\cite[\S 2.3.1]{CubicalAgda}.
We could even imagine the general quotient type defined as a HIT.
However, there are two notable differences between HITs and inductive types with conditions:
\begin{itemize}
\item Conditions give definitional equalities, while HITs talk about propositional equalities.
\item The equalities introduced by conditions need to be representable as a (partial) pattern matching,
while the path constructors of many useful HITs do not satisfy this property.
\end{itemize}

\begin{example}
\label{ex:circle}
As an example of the latter difference, consider the type of circles $\Sph1$.
It has the following constructors:
\begin{mathpar}
{\basecon:\Sph1} \and
{\loopcon:\basecon=\basecon}
\end{mathpar}
This definition is standard in topology: a circle is a path with its two endpoints glued together.
In our case, the path is $\loopcon$, and its endpoints are both $\basecon$.
We cannot turn $\loopcon$ into a condition.
\end{example}

Fortunately, recent work on cubical type theory~\cite{CCHM,CHM} provides a syntax for HITs
that could be adapted to a type theory with conditions with some lightweight additions and modifications.
We will continue the discussion on HITs later.


\addtocontents{toc}{\SkipTocEntry}
\subsection*{Contributions} 

The following list also serves as an outline of the paper.
\begin{itemize}
\item We've motivated the design of conditions by quotients (in~\cref{sub:intro-int})
and HITs (in~\cref{sub:intro-hits}) in~\cref{sec:intro}.
\item We further exemplify conditions by introducing more examples (in~\cref{sub:examples}),
discussing relevant type checking (in~\cref{sub:preserve}) and the elimination principles (in~\cref{sub:cond-elim},
motivating coconditions~\cref{ssub:coconditions}),
and its applications with an additional primitive interval type (in~\cref{sub:cubical}) in~\cref{sec:examples}.
\item We formalize the core language (in~\cref{sub:core}), its reduction rules (in~\cref{sub:ops}),
its typing rules (in~\cref{sub:typing-terms}), the confluence-checking algorithm (in~\cref{sub:confl},
the soundness~\cref{thm:soundness} of the algorithm) of a type theory with conditions in~\cref{sec:formalization}.
In~\cref{rem:other-features}, we discuss an implementation method for indexed types
that does not depend on unification-based dependent pattern matching.
\item We discuss (co)conditions from a metatheoretical perspective (in~\cref{sec:metatheory}).
\item We conclude this paper by discussing its implementations (in~\cref{sub:impl}), related work
(in~\cref{sub:rel-work}) and potential future work (in~\cref{sub:future-work}) in~\cref{sec:conclusion}.
\end{itemize}

\addtocontents{toc}{\SkipTocEntry}
\subsection*{Acknowledgements} 

We are grateful to the anonymous reviewers from POPL, ICFP, and ESOP for their valuable comments and suggestions.
Since written, we have revised this paper a couple of times according to their feedback.

\section{Programming with conditions}%
\label{sec:examples}

In this section, we show some interesting examples of conditions, coconditions
(similar to the notion of conditions, but attached to record fields),
and their type checking process.

\subsection{Naturals with infinity} 
\label{sub:examples}
We could define natural numbers with an infinity value as an inductive type with conditions.
We think of infinity as a number whose successor is itself, so the natural number type with
an infinity value can be defined as an inductive type with the following constructors:

\begin{itemize}
\item A zero constructor and an infinity constructor.
\item A constructor generating the successor of a provided number.
\end{itemize}

A condition could be used to reflect the property of the infinity value,
leading to this definition:

\lessSpace{-1.2}
\begin{align*}
&\kwdata~\NatInf:\UU \\[-0.3em]
& \mid\zerocon \\[-0.3em]
& \mid \cons{∞} \\[-0.3em]
& \mid\succon~\NatInf~
  \kwwith~\cons{∞}⇒\cons{∞}
\end{align*}

The condition makes any successor of the infinity value be itself.
We could define the arithmetic operations like addition and multiplication in the obvious ways,
and prove some properties of $\NatInf$, such as the commutativity/associativity of addition and multiplication,
$(x:\NatInf) → ∞+x=∞$, $∞+∞=∞$.
The definitions and proofs are omitted.

A similar type in \CubicalAgda{} is available on GitHub\footnote{
  \url{https://github.com/agda/cubical/blob/master/Cubical/HITs/InfNat/Base.agda}},
but it has a propositional equality rather than a definitional equality.
Some proofs about the \CubicalAgda{} version of this type\footnote{
  \url{https://github.com/agda/cubical/blob/master/Cubical/HITs/InfNat/Properties.agda}}
require Kan operations, but the equivalent theorems for our $\NatInf$ does not need them
as we do not have (syntactically nontrivial) paths in $\NatInf$ at all.

Conatural numbers~\cite[Appendix]{Copat} --- the coinductive version of the natural number type ---
is another natural number type that allows an infinity value.
The idea is that we define a conatural not by its \textit{canonical value},
but by the \textit{action} to find its predecessor.
So, zero is defined as a conatural who does not have a predecessor,
three is defined as a conatural whose predecessor is two,
and infinity is defined as a conatural whose predecessor is itself.

In Agda, the definition of conatural numbers consists of two mutually recursive types
\AgdaDatatype{Conat} and \AgdaDatatype{$\eta$Conat} (note that the coloring of Agda
snippets (types and functions are in blue, constructors are in green, and fields are in pink)
is different from the coloring of other snippets (types are in green, functions are in blue,
constructors are in purple, and fields are in aquamarine)):

\begin{code}[hide]%
\>[0]\AgdaSymbol{\{-\#}\AgdaSpace{}%
\AgdaKeyword{OPTIONS}\AgdaSpace{}%
\AgdaPragma{--cubical}\AgdaSpace{}%
\AgdaPragma{--guardedness}\AgdaSpace{}%
\AgdaSymbol{\#-\}}\<%
\\
\>[0]\AgdaKeyword{open}\AgdaSpace{}%
\AgdaKeyword{import}\AgdaSpace{}%
\AgdaModule{Cubical.Core.Everything}\<%
\\
\>[0]\AgdaKeyword{mutual}\<%
\end{code}

\begin{figure}[h]
\begin{minipage}{0.5\linewidth}
\centering
\begin{code}%
\>[0][@{}l@{\AgdaIndent{1}}]%
\>[2]\AgdaKeyword{data}\AgdaSpace{}%
\AgdaDatatype{μConat}\AgdaSpace{}%
\AgdaSymbol{:}\AgdaSpace{}%
\AgdaPrimitive{Type₀}\AgdaSpace{}%
\AgdaKeyword{where}\<%
\\
\>[2][@{}l@{\AgdaIndent{0}}]%
\>[4]\AgdaInductiveConstructor{zero}%
\>[10]\AgdaSymbol{:}\AgdaSpace{}%
\AgdaDatatype{μConat}\<%
\\
\>[4]\AgdaInductiveConstructor{suc}%
\>[10]\AgdaSymbol{:}\AgdaSpace{}%
\AgdaRecord{Conat}\AgdaSpace{}%
\AgdaSymbol{→}\AgdaSpace{}%
\AgdaDatatype{μConat}\<%
\end{code}
\end{minipage}\begin{minipage}{0.5\linewidth}
\centering
\begin{code}%
\>[2]\AgdaKeyword{record}\AgdaSpace{}%
\AgdaRecord{Conat}\AgdaSpace{}%
\AgdaSymbol{:}\AgdaSpace{}%
\AgdaPrimitive{Type₀}\AgdaSpace{}%
\AgdaKeyword{where}\<%
\\
\>[2][@{}l@{\AgdaIndent{0}}]%
\>[4]\AgdaKeyword{coinductive}\<%
\\
\>[4]\AgdaKeyword{field}\AgdaSpace{}%
\AgdaField{force}\AgdaSpace{}%
\AgdaSymbol{:}\AgdaSpace{}%
\AgdaDatatype{μConat}\<%
\end{code}
\end{minipage}
\end{figure}
\lessSpace{-0.8}

\begin{code}[hide]%
\>[0]\AgdaKeyword{open}\AgdaSpace{}%
\AgdaModule{Conat}\<%
\end{code}

In the definition of \AgdaDatatype{Conat}, the predecessor action is called \AgdaField{force}.
The infinity value of the \AgdaDatatype{Conat} type looks like the following:

\lessSpace{-0.8}
\begin{center}
\begin{code}%
\>[0]\AgdaFunction{∞}\AgdaSpace{}%
\AgdaSymbol{:}\AgdaSpace{}%
\AgdaRecord{Conat}\<%
\\
\>[0]\AgdaFunction{∞}\AgdaSpace{}%
\AgdaSymbol{.}\AgdaField{force}\AgdaSpace{}%
\AgdaSymbol{=}\AgdaSpace{}%
\AgdaInductiveConstructor{suc}\AgdaSpace{}%
\AgdaFunction{∞}\<%
\end{code}
\end{center}
\lessSpace{-0.8}

However, we cannot inspect if a given conatural number is equal to \AgdaFunction{$∞$},
as the decision will require finding the predecessor of a given conatural number infinity times.
With the definition based on conditions, we can use pattern matching to determine if a
given instance of $\NatInf$ is $\cons{∞}$.


\subsection{Preservation of definitional equalities} 
\label{sub:preserve}
We require the definitional equalities to be preserved in functions.
That is to say, given terms $f, a, b$ where $a$ and $b$ are definitionally equal,
we require that $f~a$ and $f~b$ are also definitionally equal.
This requirement is relevant to the \textit{confluence} property of type theories,
which says that the choice of reduction strategy does not affect the reduction result.
The presence of conditions introduces user-defined definitional equalities between terms that
are distinct in other type theories, so we need to make sure that such
terms are indistinguishable under definitional equalities.

We can rephrase the ``preservation of definitional equalities'' property in the following way:
for an inductive type with conditions $\Dat D$, its instance $a$ which may reduce to $b$ due to
conditions, and a function $f$ from $\Dat D$, we have $f~a ≡ f~b$ a definitional equality.

\begin{example}
Consider the predecessor function of integers $\func{pred} : \Int→\Int$ for $\Int$
(we define functions using \textit{case trees} to avoid dealing with nested pattern matchings
and first-match semantics, provided the existence of a well-developed algorithm~\cite{DepPM}
translating general dependent pattern matchings to case trees):

\begin{center}
$\func{pred}~x ⇒ \kwcase_x \left\{
\begin{array}{l}
  \negcon~n⇒ \negcon~(\succon~n)\\
  \poscon~n⇒ \kwcase_n \left\{
  \begin{array}{l}
    \zerocon⇒ \negcon~(\succon~\zerocon)\\
    \succon~m⇒ \poscon~m
  \end{array}
  \right\}
\end{array}
\right\}$
\end{center}

Since $\negcon~\zerocon$ and $\poscon~\zerocon$ are definitionally equal,
we need to make sure that $\func{pred}~(\negcon~\zerocon)$ and $\func{pred}~(\poscon~\zerocon)$
are definitionally equal, too.
The type checker should perform a conversion check between these two values.
In the $\func{pred}$ case, both of them reduce to $\poscon~(\succon~\zerocon)$,
so this function is accepted.
\end{example}

Such conversion checks will be performed after the type checks and termination checks
for every function from an inductive type with conditions.
We will refer to this conversion check as the \textit{confluence check}.


\subsection{Elimination principles} 
\label{sub:cond-elim}
Each inductive type definition generates several typing rules around it.
The definition of $\Nat$ in~\cref{sub:intro-int} generates the following typing rules,
assuming the existence of structurally recursive functions:

\lessSpace{-1.2}
\begin{mathpar}
\inferrule{Γ⊢}{Γ⊢\zerocon:\Nat}\textsc{IntroZero} \and
\inferrule{Γ⊢ n:\Nat}{Γ⊢\succon~n:\Nat}\textsc{IntroSuc} \and
\inferrule{Γ,y:\Nat⊢ A:\UU \\ Γ,y:\Nat⊢ a:A[\zerocon/y] \\
 Γ,y:\Nat,x:\Nat ⊢ b:A[\succon~x/y]}
{Γ, y:\Nat ⊢\kwcase_y {\left\{
\begin{array}{l}
  \zerocon⇒ a \\
  \succon~x⇒ b
\end{array}
\right\}} : A}\textsc{ElimNat}
\end{mathpar}

We could \textit{internalize} all of these rules as functions in a dependent type theory
by taking the conclusion of each rule as the function body and the premises as its parameters.
The \textsc{ElimNat} rule could be internalized as:

\lessSpace{-1.2}
\begin{align*}
&\func{elimNat} : (A:\Nat→\UU)~(a:A~\zerocon)~(b:(x:\Nat)→ A~(\succon~x))~(y:\Nat)→ A~y \\
&\func{elimNat}~A~a~b~y ⇒ \kwcase_y \left\{
\begin{array}{l}
  \zerocon⇒ a \\
  \succon~x⇒ b~x
\end{array}
\right\}
\end{align*}

The general case is also possible since the premises are always representable as parameters.
However, the presence of conditions will break this situation due to the definitional equality
requirement discussed in~\cref{sub:preserve}.
The definition of $\Int$ in~\cref{sub:intro-int} generates the following rules:

\lessSpace{-1.2}
\begin{mathpar}
\inferrule{Γ⊢ n:\Nat}{Γ⊢\poscon~n:\Int}\textsc{IntroPos} \and
\inferrule{Γ⊢ n:\Nat}{Γ⊢\negcon~n:\Int}\textsc{IntroNeg} \and
\inferrule{Γ,y:\Int⊢ A:\UU \\
 Γ,y:\Int,x:\Nat ⊢ a:A[\poscon~x/y] \\\\
 Γ,y:\Int,x:\Nat ⊢ b:A[\negcon~x/y] \\
 Γ ⊢ a[\zerocon/x] ≡ b[\zerocon/x]}
{Γ, y:\Int ⊢\kwcase_y {\left\{
\begin{array}{l}
  \poscon~x⇒ a \\
  \negcon~x⇒ b
\end{array}
\right\}} : A}\textsc{ElimInt}
\end{mathpar}

The definitional equality requirement, \fbox{$Γ ⊢ a[\zerocon/x] ≡ b[\zerocon/x]$},
is not directly expressible with the intensional equality type\footnote{The extensional equality type
breaks the decidability of type checking, and is too general.}.
Fortunately, we do not need the full-blown strict equality type to encode the definitional equality
introduced by the conditions. Instead, we use \textit{coconditions}.

\subsubsection{Records with coconditions.} 
\label{ssub:coconditions}
Coconditions are partial pattern matchings in record fields. The duality makes sense in the following ways:

\begin{itemize}
\item Conditions add user-defined reductions to the introduction rules (also known as constructors) of inductive types,
 while coconditions add them to the elimination rules (also knowns as fields) of records.
\item The definitional equalities defined by conditions are preserved in the elimination of inductive types,
 while those defined by coconditions are preserved in the introduction of records.
\end{itemize}

\begin{example}
We exemplify coconditions by constructing the elimination principle of $\Int$.
Consider the following record, where the field $\elim{posElim}$ corresponds to the term
$a$ in \textsc{ElimInt} and $\elim{negElim}$ corresponds to $b$:

\lessSpace{-1.5}
\begin{align*}
&\kw{record}~\Dat{IntElimData}~(A:\Int→\UU):\UU \\[-0.3em]
& \mid \elim{posElim}:(x:\Nat)→ A~(\poscon~x) \\[-0.3em]
& \mid \elim{negElim}:(x:\Nat)→ A~(\negcon~x)~\kwwith~\zerocon ⇒ \elim{posElim}~\zerocon
\end{align*}

The cocondition $\kwwith~\zerocon ⇒ \elim{posElim}~\zerocon$ requires
each instance $k$ of $\Dat{IntElimData}$ to have the definitional equality between
$k.\elim{posElim}~\zerocon$ and $k.\elim{negElim}~\zerocon$.
\end{example}

Now, we can internalize the elimination principle of $\Int$ with $\Dat{IntElimData}$:

\lessSpace{-1.5}
\begin{align*}
&\func{elimInt} : (A:\Int→\UU)~(k:\Dat{IntElimData}~A)~(y:\Int)→ A~y \\
&\func{elimInt}~A~k ⇒ \kwcase_y \left\{
\begin{array}{l}
  \poscon~x⇒ k.\elim{posElim}~x \\
  \negcon~x⇒ k.\elim{negElim}~x
\end{array}
\right\}
\end{align*}

The definitional equality between $(k.\elim{posElim}~x)[\zerocon/x]$ and
$(k.\elim{negElim}~x)[\zerocon/x]$ is exactly the cocondition in $\Dat{IntElimData}$.
On the other hand, we could internalize the introduction rules of records with coconditions
using inductive type with conditions. For instance, the introduction rule of $\Dat{IntElimData}$
could be defined using $\Int$. The definition is omitted.

\subsection{Path type and HITs.} 
\label{sub:cubical}
The most interesting application of conditions and coconditions is to simulate the cubical syntax of
the path type and higher inductive types in homotopy type theory.
We will need a primitive type $\II$ (the \textit{interval} type
inspired by cubical type theories),
and we will be able to define some interesting types with it.

\subsubsection{The interval type $\II$.} 
\label{ssub:interval}
Intuitively, this type corresponds to the closed interval $[0,1]\in\mathbb R$,
with a constructor $\lcon$ being the left endpoint $0$
and another constructor $\rcon$ being the right endpoint $1$.
Since closed intervals are path-connected\footnote{We have not defined $\Path$ yet.
We are only talking about the homotopical intuition.}, they are homotopically equivalent
to single-point spaces. Type theoretically speaking, $\II$,
while having two \textit{distinct} constructors, has to be indistinguishable from the \textit{unit type}.
To achieve this, we take advantage of the fact that $\II$ is not an inductive type ---
as long as there is no ``pattern matching'' for $\II$, it will not become the boolean type.
The fact that $\II$ is the unit type is reflected in its elimination rule.
There are many different versions of the elimination rule for $\II$
in various type theories.

For simplicity, we choose the following definition used
in the Arend proof assistant~\cite{Arend} (the underlying type theory
is known as HoTT-I~\cite{HoTT-I}),
which is similar but different from cubical type theories:
\lessSpace{-0.5}
\begin{mathpar}
\inferrule{Γ⊢ A:\II→\UU \\ Γ⊢ i:\II \\ Γ⊢ a:A~\lcon}
{Γ⊢\coe_A^i(a):A~i}\textsc{ElimI}
\end{mathpar}
It comes with the following reduction rules\footnote{The name \textsc{Regularity} is inspired
from the definition in~\cite[Specification 4.28]{XTT} and the
\textsc{Coercion Regularity} rule in~\cite[\S 2.1.1]{XTT}.
This name is used earlier in~\cite[Acknowledgements]{CCHM}.}:
\lessSpace{-0.5}
\begin{mathpar}
\inferrule{Γ⊢ A:\UU \\ Γ⊢ a:A}{Γ⊢ \coe_{\lam i A}^i(a) ≡ a}\textsc{Regularity} \and
\inferrule{Γ⊢ A:\II→\UU \\ Γ⊢ a:A}{Γ⊢ \coe_A^\lcon(a) ≡ a}\textsc{CoeLeft}
\end{mathpar}

The semantical interpretation of this elimination rule is elaborated in~\cite[\S 2]{HoTT-I}.
In this paper, we only need an intuitive understanding of $\coe$:
Since $\II$ is the unit type, families indexed by $\II$ should be constant.
So, given any such family $A$, the instances of $A~\lcon$ can be coerced to $A~i$ for any $i$.
Cartesian cubical type theory~\cite[\S 1.2]{CCTT} defines $\coe$ with a more general type
(we have modified some syntactical details to fit the convention of
this paper in order to make the type differences clear), allowing coercion between arbitrary intervals:
\begin{mathpar}
\inferrule{Γ⊢ A:(r:\II)→\UU \\
Γ⊢ r_1:\II \\ Γ⊢ r_2:\II \\ Γ⊢ a:A~r_1}{Γ⊢\coe_A^{r_1\rightsquigarrow r_2}(a):A~r_2}
\end{mathpar}

Some variants of cubical type theory define more operators such as
$\neg,\wedge,\vee$~\cite{CCHM}, but these operators can be defined in terms of $\coe$
using the constructions in~\cite[Appendix 1]{CubicalSets}.
In this paper, we will use the $\squeeze{}{}$ operator. Intuitively,
it finds the interval closer to $\lcon$ (the ``minimum'' of its arguments).
Type theoretically, it satisfies the following computation rules:
\begin{align*}
\squeeze\lcon i &≡ \lcon & \squeeze i\lcon &≡ \lcon \\
\squeeze\rcon i &≡ i & \squeeze i\rcon &≡ i
\end{align*}

\subsubsection{The $\Path$ type.} 
\label{ssub:path}
This type corresponds to the notion of \textit{path space} in topology.
In topology, a path in space $X$ is defined as a continuous map from the closed unit interval $[0,1]$
(which we already defined in~\cref{ssub:interval}) to $X$, so we could imitate this definition in a type theory.
The typing and computation rules of $\Path$ are listed below
(the mapping corresponds to the $t$ variable in the introduction rule \textsc{PathI}).
Note that this is not needed by our type theory, we are instead going to show its encoding in terms of coconditions.
\begin{mathpar}
\inferrule{Γ⊢ A:\II→\UU \\ Γ⊢ a:A~\lcon \\ Γ⊢ b:A~\rcon}
{Γ⊢ \Path~A~a~b:\UU}\textsc{PathForm} \and
\inferrule{Γ⊢ p:\Path~A~a~b \\ Γ⊢ i:\II}{Γ⊢ p~i:A~i}\textsc{PathE} \and
\inferrule{Γ, i:\II⊢ t:A~i}
{Γ⊢\plam i t :\Path~A~(t[\lcon/i])~(t[\rcon/i])}\textsc{PathI} \and
\inferrule{Γ⊢ p:\Path~A~a~b}{Γ⊢ p~\pappSym~\lcon ≡ a}\textsc{PathL} \and
\inferrule{Γ⊢ p:\Path~A~a~b}{Γ⊢ p~\pappSym~\rcon ≡ b}\textsc{PathR} \and
\inferrule{Γ, i:\II⊢ t:A~i \\ Γ⊢ j:\II}{Γ⊢ (\plam i t)~\pappSym~j ≡ t[j/i]}\textsc{PathAp}
\end{mathpar}
An instance of $\Path~A~a~b$ is similar to a wrapper of a function of type $\II→ A$, 
satisfying the definitional equalities in the conclusions of \textsc{PathL} and \textsc{PathR}.
Before showing the encoding of $\Path$ in terms of coconditions, we first show an encoding of
the intensional $\Dat{Id}$ type in terms of $\Path$.\footnote{
  We assume that the intensional $\Dat{Id}$ type is well-known in the literature,
  but not the $\Path$ type. So we offer simple examples to help intuiting the idea.}
\begin{itemize}
\item Formation. We define \fbox{$\Dat{Id}~A~a~b$} as \fbox{$\Path~(\lam i A)~a~b$}.
\item Introduction. We define \fbox{$\cons{refl}_a:\Dat{Id}~A~a~a$} as \fbox{$\plam i a:\Path~(\lam i A)~a~a$}.
\item Elimination. We assume the following version of the $\func J$ rule:
\lessSpace{-0.5}
\begin{mathpar}
\inferrule{Γ⊢ A:\UU \\ Γ⊢ x:A \\ Γ⊢ y:A \\\\
 Γ⊢ P:(y':A)→ \Dat{Id}~A~x~y'→\UU \\
Γ⊢ p:P~x~\cons{refl}_x \\ Γ⊢ q:\Dat{Id}~A~x~y}
{Γ⊢ \func J~P~p~q:P~y~q}
\end{mathpar}
We define \fbox{$\func J~P~p~q$} as
\fbox{$\coe_{\lam i{P~(q~\pappSym~i)~(\plam j {q~\pappSym~\squeeze i j})}}^\rcon(p)$}.
\end{itemize}
\begin{lem}[$J\beta$]
\label{lem:regularity}
The $\Path$ encoding satisfies the computation
$\func J~P~p~\cons{refl}_x \mapsto p$.
\end{lem}
\begin{proof}
\lessSpace{-1.5}
\begin{align*}
 & \func J~P~p~\cons{refl}_x && \text{LHS} \\
 \mapsto~&\coe_{\lam i{P~(\cons{refl}_x~\pappSym~i)~(\plam j {\cons{refl}_x~\pappSym~\squeeze i j})}}^\rcon(p)
  && \text{expansion of}~\func J \\
 \mapsto~&\coe_{\lam i{P~((\plam j x)~\pappSym~i)~(\plam j {(\plam k x)~\pappSym~\squeeze i j})}}^\rcon(p)
  && \text{expansion of}~\cons{refl}_x \\
 \mapsto~&\coe_{\lam i{P~x~(\plam j x)}}^\rcon(p) && \text{by}~\textsc{PathAp} \\
 \mapsto~&p && \text{by}~\textsc{Regularity}
\end{align*}
\lessSpace{-1.2}
\end{proof}
\begin{cor}
\label{cor:path-id}
The intensional identity type has a sound encoding using the $\Path$ type.
\end{cor}

\subsubsection{$\Path$ is a record with coconditions.} 
\label{ssub:path-cocond}
Observing the rules in~\cref{ssub:path}, we may realize that
$\Path$ is very similar to a one-field record with coconditions.
We mentioned before that paths are similar to wrappers of functions from $\II$,
and we can realize this claim with coconditions. Below is a pseudo definition:
\lessSpace{-0.5}
\begin{align*}
&\kw{record}~\Path~(A:\II→\UU)~(a:A~\lcon)~(b:A~\rcon):\UU \\[-0.3em]
& \mid \pappSym:(i:\II)→ A~i~\kwwith{\left\{\begin{array}{l}
  \lcon⇒ a \\
  \rcon⇒ b
 \end{array}\right\}}
\end{align*}
The above definition allows $\plam i a$ to be written as $\kw{record}~\set{\pappSym⇒ (\lam i a)}$
and path application $p~\pappSym~i$ to be written as $\papp p i$.
For convenience, we will denote \fbox{$\Path~(\lam i A)~a~b$} as \fbox{$a=_A b$}
when $i$ is not free in $A$. We also define the following auxiliary function:
\begin{align*}
&\pcon:(f:(i:\II)→ A~i) : \Path~A~(f~\lcon)~(f~\rcon) \\
&\pcon~f⇒ \kw{record}~\set{\pappSym ⇒ f}
\end{align*}
The definition of $\Path$ with coconditions is a pseudo definition because it uses pattern matching on $\II$
which does not exist as mentioned in~\cref{ssub:interval}.
To make sense of this definition, we create a ``backdoor'' for pattern matching
on $\II$ in (co)conditions.

The essential reason why there is no pattern matching on $\II$ is that the closed interval $[0,1]$
in $\mathbb R$ has infinitely many points, so $\II$ has infinitely many instances.
We can imagine a ``pattern matching'' for $\II$ with three cases --- for $0$, $1$ and the open interval $(0,1)$.
This pattern matching is not expressible in type theory, but we can have the \textit{partial} version of it.
In other words, the cocondition $\kwwith{\left\{\begin{array}{l}
  \lcon⇒ a \\
  \rcon⇒ b
\end{array}\right\}}$ is just a partial function of type $(i:\II)→ A~i$,
where the case for the open interval $(0,1)$ is missing.

Up to now, we have a dependent type theory with (co)conditions and four
primitives $\II, \lcon, \rcon, \coe$ and pattern matching for $\II$ in conditions,
which is powerful enough to encode the intensional $\Dat{Id}$ type.
The $\Path$ type is more general than the $\Dat{Id}$ type,
as the extensionality principle of the function type holds:
\begin{align*}
&\func{funExt}:((a:A)→ f~a=_B g~a)→ f=_{A→ B} g \\
&\func{funExt}~p⇒{\pconp{\lam{i~a}{\papp{(p~a)}i}}}
\end{align*}
We can also prove the congruence lemma without using $\func J$:
\begin{align*}
&\func{pmap}:(f:A→ B)→ a=_A b→ f~a=_B f~b \\
&\func{pmap}~f~p⇒{\pconp{\lam{i}{f~(\papp p i)}}}
\end{align*}

\subsubsection{Cubical syntax of HITs.} 
\label{ssub:cits}
We refer to the cubical syntax for HITs in homotopy type theory as ``CITs''.
HITs are inductive types with the following properties:
\begin{itemize}
\item They have constructors.
Point constructors are similar to ordinary inductive types, while path constructors
specify nontrivial identity relations of the instances of HITs.
\item They have elimination principles, mapping the points and paths to other spaces.
\item The path constructors must be preserved in the elimination
because homotopy type theory wants every function to be continuous.
\end{itemize}
Recall~\cref{ex:circle}. There are two constructors for $\Sph1$,
a point constructor $\basecon$ and a path constructor $\loopcon$.
All functions from $\Sph1$ must map the $\loopcon$ constructor to a loop in the codomain space.
The \CubicalAgda{} definition of $\Sph1$ looks like the following:

\begin{code}[hide]%
\>[0]\AgdaSymbol{\{-\#}\AgdaSpace{}%
\AgdaKeyword{OPTIONS}\AgdaSpace{}%
\AgdaPragma{--cubical}\AgdaSpace{}%
\AgdaSymbol{\#-\}}\<%
\\
\>[0]\AgdaKeyword{open}\AgdaSpace{}%
\AgdaKeyword{import}\AgdaSpace{}%
\AgdaModule{Cubical.Core.Everything}\<%
\\
\>[0]\AgdaKeyword{variable}\AgdaSpace{}%
\AgdaGeneralizable{A}\AgdaSpace{}%
\AgdaGeneralizable{B}\AgdaSpace{}%
\AgdaSymbol{:}\AgdaSpace{}%
\AgdaPrimitive{Set}\<%
\end{code}

\lessSpace{-0.8}
\begin{center}
\begin{code}%
\>[0]\AgdaKeyword{data}\AgdaSpace{}%
\AgdaDatatype{ℝ¹}\AgdaSpace{}%
\AgdaSymbol{:}\AgdaSpace{}%
\AgdaPrimitive{Type₀}\AgdaSpace{}%
\AgdaKeyword{where}\<%
\\
\>[0][@{}l@{\AgdaIndent{0}}]%
\>[2]\AgdaInductiveConstructor{base}%
\>[8]\AgdaSymbol{:}\AgdaSpace{}%
\AgdaDatatype{ℝ¹}\<%
\\
\>[2]\AgdaInductiveConstructor{loop}%
\>[8]\AgdaSymbol{:}\AgdaSpace{}%
\AgdaInductiveConstructor{base}\AgdaSpace{}%
\AgdaOperator{\AgdaFunction{≡}}\AgdaSpace{}%
\AgdaInductiveConstructor{base}\<%
\end{code}
\end{center}
\lessSpace{-0.8}

During type checking, \CubicalAgda{} translates \AgdaInductiveConstructor{loop} into a constructor of type
$\II→\Sph1$ to make sense of it as an inductive constructor.
\textit{Boundaries} are also generated, containing the fact that \AgdaInductiveConstructor{loop} reduces
to \AgdaInductiveConstructor{base} when applied by $\lcon$ or $\rcon$.
In \CubicalAgda{}, the endpoints of the path constructors are denoted as $\AgdaInductiveConstructor{i0}$ and $\AgdaInductiveConstructor{i1}$.
The boundary information of \AgdaInductiveConstructor{loop} is denoted as
$\AgdaInductiveConstructor{loop}~\emptyset~[i\mid(\AgdaInductiveConstructor{base},\AgdaInductiveConstructor{base})]$
according to~\cite[Fig. 2]{CubicalAgda}.
The elimination of CITs can be accomplished simply by pattern matching thanks to the translation.
\CubicalAgda{} performs a \textit{boundary check} (\cite[\S 4.3, rule \textsc{CtSplitConHIT}]{CubicalAgda})
to ensure that pattern matching on CITs preserves the path constructors.
Taking \AgdaDatatype{$\mathbb S^1$} as an example, given a pattern matching $\left\{\begin{array}{l}
  \AgdaInductiveConstructor{base}⇒ a \\
  \AgdaInductiveConstructor{loop}~i⇒ b
\end{array}\right\}$, there must be two definitional equalities \fbox{$b[\AgdaInductiveConstructor{i0}/i]≡ a$} and
\fbox{$b[\AgdaInductiveConstructor{i1}/i]≡ a$}.

In short, \AgdaDatatype{$\mathbb S^1$} is an inductive type with two constructors:
$\AgdaInductiveConstructor{base}:\AgdaDatatype{$\mathbb S^1$}$ and
$\AgdaInductiveConstructor{loop}:\AgdaDatatype{I}→\AgdaDatatype{$\mathbb S^1$}$,
where the second constructor has some special reduction rules (the boundaries)
that must be preserved in pattern matching.
This naturally inspires the following definition with conditions:

\lessSpace{-1.3}
\begin{align*}
&\kwdata~\Sph1:\UU \\[-0.3em]
& \mid\basecon \\[-0.3em]
& \mid\loopcon~\II~\kwwith~\left\{
\begin{array}{l}
 \lcon⇒\basecon \\
 \rcon⇒\basecon
\end{array}\right\}
\end{align*}

It has all the properties of the \CubicalAgda{} version of \AgdaDatatype{$\mathbb S^1$},
including the reduction rules of the constructor \AgdaInductiveConstructor{loop}
and the preservation of these reduction rules in pattern matching.

We can observe the major difference between inductive types with conditions and CITs:
the former uses the computation rules to derive the endpoints,
and the latter uses the endpoints to derive computation rules.
The endpoints are described using a denser but harder to read syntax,
while computation rules are easier to read but slightly more boilerplate.
This difference is more pronounced when the dimension of the path constructors grows higher,
as the face constructor of $\Tori$. We would write out the computation rules directly
with conditions:

\lessSpace{-1.8}
\begin{align*}
&\kwdata~\Tori:\UU \\[-0.3em]
& \mid\pointcon \\[-0.3em]
& \mid\linecon 1~\II~\kwwith~\left\{\begin{array}{l}
  \lcon⇒\pointcon \\
  \rcon⇒\pointcon
\end{array}\right\} \\[-0.3em]
& \mid\linecon 2~\II~\kwwith~\left\{\begin{array}{l}
  \lcon⇒\pointcon \\
  \rcon⇒\pointcon
\end{array}\right\} \\[-0.3em]
& \mid\facecon~(i~j:\II)~\kwwith~\left\{\begin{array}{l}
 \lcon,i⇒\linecon 2~i \\
 \rcon,i⇒\linecon 2~i \\
 i,\lcon⇒\linecon 1~i \\
 i,\rcon⇒\linecon 1~i
\end{array}\right\}
\end{align*}

\begin{code}[hide]%
\>[0]\AgdaSymbol{\{-\#}\AgdaSpace{}%
\AgdaKeyword{OPTIONS}\AgdaSpace{}%
\AgdaPragma{--cubical}\AgdaSpace{}%
\AgdaSymbol{\#-\}}\<%
\\
\>[0]\AgdaKeyword{open}\AgdaSpace{}%
\AgdaKeyword{import}\AgdaSpace{}%
\AgdaModule{Cubical.Core.Everything}\<%
\\
\>[0]\AgdaKeyword{variable}\AgdaSpace{}%
\AgdaGeneralizable{A}\AgdaSpace{}%
\AgdaGeneralizable{B}\AgdaSpace{}%
\AgdaSymbol{:}\AgdaSpace{}%
\AgdaPrimitive{Set}\<%
\end{code}

It is slightly more difficult to tell the computation rules from
the \texttt{Cubical Agda} version of \AgdaInductiveConstructor{face}:

\lessSpace{-0.8}
\begin{center}
\begin{code}%
\>[0]\AgdaKeyword{data}\AgdaSpace{}%
\AgdaDatatype{Torus}\AgdaSpace{}%
\AgdaSymbol{:}\AgdaSpace{}%
\AgdaPrimitive{Type₀}\AgdaSpace{}%
\AgdaKeyword{where}\<%
\\
\>[0][@{}l@{\AgdaIndent{0}}]%
\>[2]\AgdaInductiveConstructor{point}%
\>[9]\AgdaSymbol{:}\AgdaSpace{}%
\AgdaDatatype{Torus}\<%
\\
\>[2]\AgdaInductiveConstructor{line1}%
\>[9]\AgdaSymbol{:}\AgdaSpace{}%
\AgdaInductiveConstructor{point}\AgdaSpace{}%
\AgdaOperator{\AgdaFunction{≡}}\AgdaSpace{}%
\AgdaInductiveConstructor{point}\<%
\\
\>[2]\AgdaInductiveConstructor{line2}%
\>[9]\AgdaSymbol{:}\AgdaSpace{}%
\AgdaInductiveConstructor{point}\AgdaSpace{}%
\AgdaOperator{\AgdaFunction{≡}}\AgdaSpace{}%
\AgdaInductiveConstructor{point}\<%
\\
\>[2]\AgdaInductiveConstructor{face}%
\>[9]\AgdaSymbol{:}\AgdaSpace{}%
\AgdaPostulate{PathP}\AgdaSpace{}%
\AgdaSymbol{(λ}\AgdaSpace{}%
\AgdaBound{i}\AgdaSpace{}%
\AgdaSymbol{→}\AgdaSpace{}%
\AgdaInductiveConstructor{line1}\AgdaSpace{}%
\AgdaBound{i}\AgdaSpace{}%
\AgdaOperator{\AgdaFunction{≡}}\AgdaSpace{}%
\AgdaInductiveConstructor{line1}\AgdaSpace{}%
\AgdaBound{i}\AgdaSymbol{)}\AgdaSpace{}%
\AgdaInductiveConstructor{line2}\AgdaSpace{}%
\AgdaInductiveConstructor{line2}\<%
\end{code}
\end{center}
\lessSpace{-0.8}

However, the definition based on conditions is almost twice as long as the CIT definition.
Despite the syntactical difference, the two versions can be used to replicate each other.

\section{Formalization of (co)conditions}%
\label{sec:formalization}
In this section, we will show how inductive types with conditions are type-checked.
We define the rules for type-checking inductive type definitions
and confluence-checking for pattern matching on inductive types with conditions.
The type-checking of records with coconditions and the confluence-checking of
their introduction are omitted from this paper because they are too similar to inductive types with conditions.

\subsection{Core language syntax} 
\label{sub:core}
The term language in our type theory is presented in~\cref{fig:syntax}.

\begin{figure}[h!]
\begin{align*}
  x,y,z,i,j ::= & && \text{variable names} \\[-0.3em]
  p, q    ::= & \quad \cons c~\overline x ⇒ u && \text{constructor clause} \\[-0.3em]
  \sigma,\gamma ::= & \quad [\overline u/\overline x]  && \text{substitutions} \\[-0.3em]
  A,B,a,b,u,v ::= & \quad \func f && \text{function invocation} \\[-0.3em]
         \mid & \quad \Dat D~\overline u && \text{fully applied inductive type} \\[-0.3em]
         \mid & \quad \cons c~\overline u && \text{fully applied constructor} \\[-0.3em]
         \mid & \quad v~u && \text{application to non-definitions} \\[-0.3em]
         \mid & \quad \kwcase_x~\set{\overline p} && \text{case-split} \\[-0.3em]
         \mid & \quad (x:A)→ B && \text{$\Pi$-type} \\[-0.3em]
         \mid & \quad \II && \text{interval type} \\[-0.3em]
         \mid & \quad \coe && \text{interval elimination} \\[-0.3em]
         \mid & \quad \rcon &&  \\[-0.3em]
         \mid & \quad \lcon &&  \\[-0.3em]
         \mid & \quad \UU && \text{universe}
\end{align*}
\caption{Syntax of terms}
\label{fig:syntax}
\end{figure}

\begin{notation}
We borrow some notational convention from~\cite{SIT}:
A list of $u$ is written as $\overline u$ or $\emptyset$ when the list is empty.
Substitution of occurrences of $x$ with $v$  is written as {$[v/x]$},
while a sequentially applied list of substitutions is written as {$[\overline v/\overline x]$}.
We assume a substitution operation defined on patterns, terms, contexts, and other substitutions
(written as $u\sigma$, $p\sigma$, $Γ\sigma$ and $\sigma\gamma$).
We will have two shorthands for the $\Pi$-type. We write $A → B$ for $(x:A) → B$ if $x$ is not free in $B$,
and $\Delta → B$ for telescopic $\Pi$-types.
\end{notation}

The system is almost identical to every other dependently-typed programming language,
such as Agda~\cite{AgdaOri}, Idris~\cite{Idris}, or Lean~\cite{Lean} with some subtle differences:
\begin{enumerate}
\item To avoid dealing with universes that are unrelated to conditions,
we consider $\UU$ a well-typed type expression, and it does not have a type.
Alternatively, we can adapt a predicative hierarchy of universes.
We use universe \`a la Russell for syntactical convenience.
\item Case-split expressions only work for variables and unnested patterns,
similar to Mini-TT~\cite[\S 6.2]{MiniTT} and the case trees in Agda~\cite[\S 4, (27)]{DepPM}.
\end{enumerate}

Then, we extend the syntax to define inductive types with conditions and functions in~\cref{fig:sig-decl}.

\begin{figure}[h!]
\begin{align*}
  Γ,\Delta ::= & \quad \overline{x_i:A_i} && \text{context} \\[-0.3em]
  decl ::= & \quad \kwdata~\Dat D~\Delta~\overline{cons} && \text{inductive type} \\[-0.3em]
      \mid & \quad \kw{func}~\func f:A⇒ v && \text{function definition} \\[-0.3em]
  cons ::= & \quad \mid \cons c~\Delta~\kwwith~u && \text{constructor} \\[-0.3em]
  \Sigma ::= & \quad \overline{decl} && \text{signature}
\end{align*}
\caption{Syntax of signature and declarations}
\label{fig:sig-decl}
\end{figure}

We introduce the notion of \textit{signature}, which stores a list of declarations,
and \textit{context}, which is a telescopic list of bindings.
A constructor without conditions is equivalent to a constructor with an empty condition,
so we assume all constructors to have a condition for the convenience when it comes to formalization.
In typing rules, we omit the vertical bars in $cons$ which are just syntactical separators of constructors.

\begin{remark}
\label{rem:other-features}
We formalized a type theory without any notion of ``indexed types''.
A simple encoding of a subset of indexed types is formalized in~\cite{SIT},
which features the independence of unification-based dependent pattern matching,
known to be extremely complicated in both semantical interpretation and practical implementation.
In~\cite[\S 4.1]{SIT}, the authors claimed that if there is a built-in identity type,
all of the traditional ``indexed types'' can be encoded using the proposed syntax.
They mentioned that the syntactical convenience provided by dependent pattern matching can be superseded by
an externally implemented \textsf{rewrite} mechanism, so this encoding is not even syntactically lossy.
Fortunately, $\Path$ is such an identity type as proved in~\cref{cor:path-id}.
Since (co)conditions are independent of unification-based dependent pattern matching,
we can develop an implementation of intensional type theory
entirely without dependent pattern matching and the index-unification algorithm.
\end{remark}

\subsection{Operations on terms} 
\label{sub:ops}
Before getting into the typing rules, we need to introduce some operations on terms.

We define \fbox{$\vars\Delta$} to compute the list of variables in $\Delta$:
\begin{align*}
\vars\emptyset &:= \emptyset \\
\vars{x:A,\Delta} &:= x,\vars\Delta
\end{align*}
This operation is notated as $\hat\Delta$ in~\cite{DepPM}, but we believe that using a
function name instead of a symbol makes it easier to infer the meaning of these operations.

We define \fbox{$\replace\Delta{\overline x}$} for replacing the names in $\Delta$ with $\overline x$,
assuming $\Delta$ and $\overline x$ have the same length.
\begin{align*}
\replace\emptyset\emptyset &:= \emptyset \\
\replace{(y:A,\Delta)}{(x',\overline x)} &:= x':A,\replace{\Delta[x'/y]}{\overline x}
\end{align*}

The reduction rules of terms \fbox{$\Sigma⊢ a \mapsto b$} are defined in~\cref{fig:whnf}.
It is good to have $\eta$-conversion, but it is also unnecessary. Therefore,
we use only untyped reduction which does not use the typing context.

\begin{figure}[h!]
\RaggedRight
Elimination of the $\Pi$-type.
\begin{mathpar}
\inferrule{\Sigma⊢ u \mapsto \lam x a \\\\ \Sigma⊢ a[v/x] \mapsto b}
{\Sigma⊢ u~v \mapsto b} \and
\inferrule{\Sigma⊢ u' \mapsto \kwcase_y~\set{\overline p} \\
\Sigma⊢ v' \mapsto \cons c~\overline v \\\\
 \cons c~\overline x⇒ u \in \overline p \\
 \Sigma⊢ u[\overline v/\overline x][\cons c~\overline v/y] \mapsto a}
{\Sigma⊢ u'~v'\mapsto a} \and
\inferrule{\kw{func}~\func f~:A ⇒ u\in \Sigma}
{\Sigma⊢ \func f \mapsto u}
\end{mathpar} \\
Reduction due to conditions.
\begin{mathpar}
\inferrule{\kwdata~\Dat D~\Delta~\kwwhere~\overline{cons} \in \Sigma \\
\cons c~\Delta_{\cons c}~\kwwith~u\in \overline{cons} \\
\Sigma⊢ \overline v \mapsto \overline b \\
\Sigma⊢ u~\overline b \mapsto a}
{\Sigma⊢ \cons c~\overline v \mapsto a}
\end{mathpar}
\caption{Reduction rules}
\label{fig:whnf}
\end{figure}

We omit the \textsc{Regularity} rule and the \textsc{CoeLeft} rule
as they are already defined in~\cref{ssub:interval}.

We denote the confluence-checking algorithm as
\fbox{$\SGvdash \confluence\Delta u v{\overline p}A$}, where:
\begin{itemize}
\item $u$ is the body of a pattern matching on a constructor $\cons c$ with a condition.
\item $v$ is the condition of $\cons c$.
\item $\Delta$ is a context containing the parameters of $\cons c$.
\item $\overline p$ is the list of pattern matching clauses that $u$ belongs to.
\item $A$ is the result type of the pattern matching clauses in $\overline p$.
\end{itemize}
We postpone its definition to~\cref{sub:confl} considering its complexity and
the high relevance to its usage in the typing rules
(in particular, the \textsc{SplitCt} rule and the \textsc{SplitCond} rule in~\cref{fig:type-expr}).

\subsection{Typing rules} 
\label{sub:typing-terms}
Our type theory has the following typing judgments, all parameterized by
a signature $\Sigma$ and a context $Γ$ which are assumed to be well-formed.
Type $A$ or context $\Delta$ appearing at the right-hand side of the colons in
the judgments are also assumed to be well-formed.
\begin{itemize}
\item \fbox{$\SGvdash \Delta$} $\Delta$ is a well-formed context in $\Sigma;Γ$
\item \fbox{$\SGvdash u:A$} term $u$ has type $A$ in $\Sigma;Γ$
\item \fbox{$\SGvdash_c u:A$} $u$ is a valid condition of type $A$ in $\Sigma;Γ$
\item \fbox{$\SGvdash \overline u:\Delta$} terms $\overline u$ instantiate context $\Delta$ in $\Sigma;Γ$
\item \fbox{$\SGvdash u ≡ v:A$} terms $u$ and $v$ are equal instances of type $A$ in $\Sigma;Γ$
\end{itemize}
Each inference rule entails the well-formedness of the signature and the context.
The expected type or context which appears on the right-hand side of the colon.
We can also define \fbox{$\SGvdash \overline u=\overline v:\Delta$} in a way
similar to a combination of the last two judgments, but it is not used in any of the typing rules.

The typing rules for terms are defined in~\cref{fig:type-expr}, grouped by the relevant formation rule.
Unimportant rules are not labeled.
We omit the rule \textsc{ElimI} which is already defined in~\cref{ssub:interval}.
\begin{figure}[h!]
\RaggedRight
The rules related to the $\II$-type:
\begin{mathpar}
\inferrule{}{\SGvdash \II:\UU} \and \inferrule{}{\lcon:\II} \and
\inferrule{}{\rcon:\II}
\end{mathpar} \\
The rules related to the $\Pi$-type:
\begin{mathpar}
\inferrule{\SGvdash A:\UU \\ \Sigma;Γ,x:A⊢ B:\UU}
{\SGvdash (x:A)→ B:\UU} \and
\inferrule{\kw{func}~\func f:A\in\Sigma}{\SGvdash \func f:A} \and
\inferrule{\Sigma;Γ,x:A⊢ b:B[x/y]}{\SGvdash \lam x b : (y:A)→ B} \and
\inferrule{\SGvdash u : (x:A)→ B \\ \SGvdash v:A}{\SGvdash u~v:B[v/x]}
\end{mathpar} \\
The rules related to the inductive types (with conditions):
\begin{mathpar}
\inferrule{\kwdata~\Dat D~\Delta~\overline{cons}\in\Sigma
 \\ \SGvdash \overline u : \Delta}
 {\SGvdash \Dat D~\overline u:\UU} \and
\inferrule{\kwdata~\Dat D~\Delta~\overline{cons} \in\Sigma \\
\cons c~\Delta_{\cons c}~\kwwith~{v_{\cons c}}\in \overline{cons} \\\\
 \SGvdash \overline u:\Delta \\
 \SGvdash \overline v:\Delta_{\cons c}[\overline u/\vars\Delta]}
 {\SGvdash \cons c~\overline v:\Dat D~\overline u} \and
\inferrule{\kwdata~\Dat D~\Delta~\emptyset \in \Sigma \\
 \SGvdash \overline v:\Delta \\ \Sigma;Γ,x:\Dat D~\overline v ⊢ A:\UU}
 {\SGvdash \kwcase_x~\set{}:\Dat D~\overline v → A}\textsc{SplitEmpty} \and
\inferrule{\kwdata~\Dat D~\Delta~\overline{cons}\in\Sigma \\
 \Sigma;Γ,y:\Dat D~\overline u ⊢ A[y/z] : \UU \\
 \SGvdash \overline u : \Delta \\\\
 \left(
  \inferrule{(\cons {c_k}~\overline x⇒ u_k)\in\overline p \\
    Γ' := Γ,y:\Dat D~\overline u,\replace{\Delta_k}{\overline x} \\\\
    \Sigma;Γ'⊢ u_k:A[y/z][\cons {c_k}~\overline x/y] \\\\
    \Sigma;Γ'⊢ \confluence{\Delta_k}{v_k[\overline u/\vars\Delta]}{u_k}{\overline p}{A[y/z]}}{}
 \right)_{\forall (\cons {c_k}~\Delta_k~\kwwith~{v_k}) \in \overline{cons}}}
 {\SGvdash \kwcase_y~\set{\overline p}:(z:\Dat D~\overline u)→ A}\textsc{SplitCt}
\end{mathpar} \\
The rules for conditions:
\begin{mathpar}
\inferrule{\SGvdash a:A}{\SGvdash_c a:A} \and
\inferrule{\Sigma;Γ,j:\II ⊢ A[j/i] : \UU \\\\
\left(
  \inferrule{\cons {c_k} \in \set{\lcon,\rcon} \\\\
    \Sigma;Γ,j:\II⊢_c u_k:A[j/i][\cons {c_k}/j]}{}
\right)_{\forall (\cons {c_k}⇒ u_k) \in \overline p}}
{\SGvdash_c \kwcase_j~\set{\overline p}:(i:\II)→ A}\textsc{SplitI} \and
\inferrule{\kwdata~\Dat D~\Delta~\overline{cons}\in\Sigma \\
\Sigma;Γ,y:\Dat D~\overline u ⊢ A[y/z] : \UU \\
\SGvdash \overline u : \Delta \\\\
\left(
  \inferrule{(\cons {c_k}~\Delta_k~\kwwith~{v_k}) \in \overline{cons} \\
    Γ' := Γ,y:\Dat D~\overline u,\replace{\Delta_k}{\overline x} \\\\
    \Sigma;Γ'⊢_c u_k:A[y/z][\cons {c_k}~\overline x/y] \\\\
    \Sigma;Γ'⊢ \confluence{\Delta_k}{v_k[\overline u/\vars\Delta]}{u_k}{\overline p}{A[y/z]}}{}
\right)_{\forall (\cons {c_k}~\overline x⇒ u_k) \in \overline p}}
{\SGvdash_c \kwcase_{y}~\set{\overline p}:(z:\Dat D~\overline u)→ A}\textsc{SplitCond}
\end{mathpar} \\
The rule for convertible types.
\begin{mathpar}
\inferrule{\SGvdash a:A \\ \SGvdash A ≡ B : \UU}{\SGvdash a:B}
\end{mathpar}
\caption{Typing rules for types, terms, and conditions}
\label{fig:type-expr}
\end{figure}

The most significant (and complex) rules are the \textsc{SplitCt} and the \textsc{SplitCond} 
which takes care of case-split terms.
Both of them work similarly in the following way:
for a constructor \fbox{$\cons {c_k}~\Delta_k~\kwwith~{v_k}$}
and a pattern matching clause \fbox{$\cons {c_k}~\overline x⇒ u_k$},
we first add the binding introduced by the case-split \fbox{$y:\Dat D~\overline u$}
and the bindings introduced by the pattern \fbox{$\replace{\Delta_k}{\overline x}$} into the context,
and check the clause body $u_k$ against the result type $A[\cons {c_k}~\overline x/y]$.
After that, we perform a confluence check, which will be further discussed in~\cref{sub:confl}.
\textsc{SplitI} also deals with case-splits, but it is specific to splitting the $\II$ type,
whose constructors have neither parameters nor conditions, leading to a much simpler rule.

\begin{remark}
The rules do not eliminate redundant or duplicate patterns, as such elimination shall be done in
practical implementations of type theories with conditions.
\end{remark}

\begin{remark}
Interestingly, the only difference between \textsc{SplitCond} and \textsc{SplitCt}
is in the exchange of the quantification and a premise inside the quantification --
\fbox{$(\cons {c_k}~\overline x⇒ u_k) \in \overline p$}
and \fbox{$(\cons {c_k}~\Delta_k~\kwwith~{v_k}) \in \overline{cons}$}.

This is because we do not require coverage when the case expressions are in conditions,
so we traverse the \textit{defined} clauses and check them in \textsc{SplitCond}.
In the case of normal case-splits where coverage is required, we traverse the \textit{required}
clauses and check them in \textsc{SplitCt}.
\end{remark}

To form a valid signature $\Sigma$, we check a list of declarations sequentially and
store them into the signature, allowing us to check latter declarations against prior ones.
The notation for signature formation is \fbox{$\SGvdash$}, defined in~\cref{fig:sign-well}.

\begin{figure}
\begin{mathpar}
\inferrule{
\left.\inferrule{\SGvdash \\\\ \SGvdash \Delta}{}\right. \\
\left(
  \inferrule{
    \Delta_c := \overline{\cons{c_j} : \Delta_j → \Dat D~\vars\Delta}_{\forall j\in[1, k)}
    \\\\
  \Sigma;Γ,\Delta,\Delta_c ⊢ \Delta_k \\\\
  \Sigma;Γ,\Delta,\Delta_c ⊢_c u_k : \Delta_k → \Dat D~\vars\Delta}{}
\right)_{\forall (\cons {c_k}~\Delta_k~\kwwith~u_k) \in \overline{cons}}}
{\Sigma,\kwdata~\Dat D~\Delta~\overline{cons}; Γ⊢}\textsc{DataDecl}
\and
\inferrule{\SGvdash \\ \Sigma;Γ⊢ u : A}
{\Sigma,\kw{func}~\func f:A ⇒ u; Γ⊢}\textsc{FunDecl}
\end{mathpar}
\caption{Well-formedness of signature $\Sigma$}
\label{fig:sign-well}
\end{figure}

In \textsc{DataDecl}, the constructors are assumed to be checked sequentially,
and the type checking of the $k$-th constructor uses the type checking result
of the $j$-th constructors where $j\in[1,k)$.

\begin{lem}[Ended]
\label{lem:terminate}
Reduction based on conditions in~\cref{fig:whnf} will always terminate.
\end{lem}
\begin{proof}
Neither direct nor indirect recursion are allowed in conditions,
and the number of constructors of an inductive type is finite.
\end{proof}

\begin{remark}
\label{rem:bindings}
According to \textsc{DataDecl}, the condition $u_k$ in the $k$-th constructor of an
inductive type $\Dat D~\Delta$ has access to the following variables:
the type parameters $\Delta$, the previously defined constructors $\Delta_c$,
and the provided local bindings $Γ$.
\end{remark}

\subsection{Checking confluence} 
\label{sub:confl}
The confluence checking algorithm \fbox{$\SGvdash \confluence\Delta u v{\overline p}A$}
mentioned previously at the end of~\cref{sub:ops} is defined in~\cref{fig:confluence}.
\begin{figure}[h!]
\begin{mathpar}
\inferrule{(\cons c~\overline x ⇒ u_{\cons c}) \in \overline p \\
\Sigma;Γ⊢ v ≡ u_{\cons c}[\overline u / \overline x] : A}
{\SGvdash \confluence\emptyset{\cons c~\overline u}{v}{\overline p}A}
\textsc{ConfCon}\and
\inferrule{\Sigma;Γ,x:B ⊢ \confluence\Delta{u[x/y]}v{\overline p}A}
{\SGvdash \confluence{(x:B,\Delta)}{\lam y u}v{\overline p}A}
\textsc{ConfAbs}\and
\inferrule{\kwdata~\Dat D~\Delta_D~\overline{cons}\in\Sigma \\\\ \left(
\inferrule{}{\cons{c_i}~\Delta_i \in \overline{cons} \\ Γ' := Γ,\replace{\Delta_i}{\overline{x_i}} \\\\
\Sigma⊢ u_i[\cons{c_i}~\overline{x_i}/y] \mapsto u \\\\
\Sigma;Γ',x:B ⊢ \confluence\Delta u{v[\cons{c_i}~\overline{x_i}/x]}{\overline p}A}
\right)_{\forall (\cons{c_i}~\overline{x_i}⇒ u_i) \in \overline q}}
{\SGvdash \confluence{(x:B,\Delta)}{\kwcase_y~\set{\overline q}}v{\overline p}A}
\textsc{ConfSplit}
\end{mathpar}
\caption{Confluence checking}
\label{fig:confluence}
\end{figure}

\fbox{$\SGvdash \confluence\Delta u v{\overline p}A$} is defined
by performing induction simultaneously on $\Delta$ and the condition $u$.
This assumes that $u$ must be in constructor form when $\Delta$ is empty (\textsc{ConfCon}),
and is otherwise either in abstraction or case-split form (\textsc{ConfAbs} and \textsc{ConfSplit}).

Here is a textual explanation for the rules in~\cref{fig:confluence}:
\begin{description}
\item[\textsc{ConfCon}]
  This rule applies when the condition has run out of parameters to pattern-match against.
  Recall that $\cons c~\overline u$ is the body of the clause for $\cons{c_k}$ in $\overline p$,
  and $\cons c$ also has its own corresponding clause whose body is $v$.

  Let us remind ourselves the goal of the confluence check:
  given an expression $\cons{c_k}~\overline u:\Dat D~\Delta_D$, we can reduce it in two ways:
  (1) reduce by applying the condition of $\cons{c_k}$ and then applying a suitable clause in $\overline p$ or
  (2) just reduce by applying the clause for $\cons{c_k}$ in $\overline p$.
  In other words, assuming a function $f=\kwcase_x~\set{\overline p}$ and the reduction
  $\cons{c_k}~\overline u \mapsto \cons c~\overline v$ due to the condition of $\cons{c_k}$,
  there must be $f~(\cons c~\overline v) ≡ f~(\cons{c_k}~\overline u)$ a definitional equality.

  The premise \fbox{$\Sigma;Γ⊢ v ≡ u_{\cons c}[\overline u / \overline x] : A$}
  checks whether this definitional equality holds.
\item[\textsc{ConfAbs}]
  This rule applies when the condition does not split on a parameter.
  In this case, we substitute $w_k$ and continue the confluence check with the rest
  of the conditions.
\item[\textsc{ConfSplit}]
  This rule applies when the condition splits on a parameter.
  In this case, we just check whether the rest of the conditions are satisfied for all the splits.
\end{description}

\begin{remark}
In \textsc{SplitCond} and \textsc{SplitCt} we can invoke confluence checks before all
clauses are type-checked, while in \textsc{ConfCon} we may use any clause in $\overline p$,
including the ones that are not yet type-checked.
To address the ordering issue, we type check all clauses and then
perform their confluence checks afterward in the actual implementations.
\end{remark}

\begin{lem}
\label{lem:introduced}
A well-typed condition is generated from lambda abstractions, partial case-splits,
and applications to other constructors only.
\end{lem}
\begin{proof}
By induction on the rules in~\cref{fig:confluence}, which traverses $\Delta$
and allows only constructor application when $\Delta$ is empty.
By~\cref{lem:terminate}, the constructor applied cannot be the current constructor.
\end{proof}

The rules in~\cref{fig:confluence} use several terms in different contexts.
Now, we show that the terms are always used in proper contexts.
\begin{lem}
\label{lem:bounded}
The premise \fbox{$\Sigma;Γ⊢ v ≡ u_{\cons c}[\overline u / \overline x] : A$}
in the \textsc{ConfCon} rule is well-scoped.
\end{lem}
\begin{proof}
We only discuss the terms that moved out of the contexts where they were originally typed.
\begin{enumerate}
\item The condition is well-scoped. By~\cref{rem:bindings}, it has access to
the parameters of the inductive $\Dat D$, previous constructors, and provided local variables.
The latter two are in $Γ$ (from \textsc{DataDecl}, through \textsc{SplitCt} and \textsc{SplitCond}),
while the parameters are substituted out in \textsc{SplitCt} and \textsc{SplitCond}.
\item $\cons c~\overline u$ in \textsc{ConfCon} is the innermost term in the condition,
so it may contain bindings from abstractions, case-splits, and patterns.
The variable bound by abstractions and case-splits ($y$ in both rules) are substituted out
with the variable $x$ in \textsc{ConfAbs} (added to $Γ$) and term
$\cons{c_i}~\overline{x_i}$ (added to $Γ'$ by a \textsf{replace} operation).
\item $u_{\cons c}$ in \textsc{ConfCon} have access to $\overline x$,
which is substituted out before the conversion check.
\item $v$ comes from the term $u_k$ in \textsc{SplitCond} and \textsc{SplitCt}.
It can access $\overline x$ in these two rules,
which are added to $Γ'$ by a \textsf{replace} operation.
\end{enumerate}
\end{proof}

\begin{thm}
\label{thm:soundness}
If $\SGvdash u:\Dat D→ A, \SGvdash \cons c:\Delta→\Dat D, \SGvdash \overline v:\Delta$,
suppose $\Sigma⊢ u~(\cons c~\vars\Delta)\mapsto a',
\Sigma⊢ a'[\overline v/\vars\Delta]\mapsto a$ and
$\Sigma⊢ \overline v\mapsto \overline{b'},
\Sigma⊢ u~(\cons~\overline{b'})\mapsto b$,
then $\SGvdash a≡ b$.
\end{thm}
\begin{proof}
By case-analysis on the derivation of $u$.
We only consider the cases where a reduction rule (\cref{fig:whnf}) applies.
\begin{itemize}
\item Lambda abstractions do not distinguish terms with different heads,
so $a$ and $b$ are equal by the associativity of substitution.
\item Case-splits work in a similar way to lambda abstractions when the last
rule in~\cref{fig:whnf} is unused in the reduction from $\overline v$ to $\overline{b'}$.
Otherwise the equality reduces to the equality in~\cref{fig:confluence}, rule \textsc{ConfCon}.
\end{itemize}
\end{proof}

\section{Metatheoretical discussions} 
\label{sec:metatheory}
We believe that adding (co)conditions to any type theory will not break their
logical consistency and normalization, and the proofs are straightforward
applications of standard argument (like~\cite{CDR}).
A recent development on semantic approach towards normalization~\cite{CNorm}
may also be applicable, but we are not there yet.

By~\cref{thm:soundness}, we also believe that the reduction rules
in~\cref{fig:whnf} satisfy the global confluence property. This allows us
to choose whatever reduction strategy in practical implementations.

Our type theory uses some constructions in the HoTT-I model as described in~\cref{sub:cubical}.
Some of these constructions, including the $\coe$ eliminator of the interval type,
are defined in a way that breaks canonicity.
For example, the closed term $\coe^\rcon_{\lam i A}$ is a normal form if $i$ is free in $A$,
but it is certainly not a canonical form.

There are two ways to get canonicity back:
\begin{itemize}
\item Remove the interval type to get rid of $\coe$.
This is like trading the canonicity property with the ability to define CITs and $\Path$.
\item Instead of adding cubical features to the type theory with (co)conditions,
replace some constructions in an existing cubical type theory (such as~\cite{HIT-CCTT})
with their (co)condition-based definition and see if we can derive a canonicity theorem from that.
\end{itemize}

The latter is an ideal solution, but we did not use it for the following reasons:
\begin{itemize}
\item Cubical type theories have large rules, therefore proving its metatheoretical properties can be a pain.
\item There is not yet a cubical type theory with regularity, which is necessary for~\cref{lem:regularity}.
This is a difficult problem as described in~\cite[\S 3.4]{CSCCTT}.
\item Cubical type theories define the $\Path$ type as a primitive, while our type theory does not
require it to be a primitive. So, the ideal solution is not integrating a cubical type theory directly,
but a modified version of cubical type theory with unneeded primitives requirement removed.
We consider this adaption as future work.
\item The Kan rules (i.e.~reduction of $\coe$) on the $\Path$ type from cubical type theories
need to be generalized to allow arbitrary records with coconditions.
\end{itemize}

\section{Conclusions} %
\label{sec:conclusion}

We have discussed (co)conditions, a generalization of CITs and the cubical path type.
Although the design of (co)conditions is initially centered around interpreting constructions in homotopy type theory,
there are some other types unrelated to homotopy type theory
(like $\Int$ in~\cref{sub:intro-int} and $\NatInf$ in~\cref{sub:examples}) that can be defined with conditions.
We have formalized a dependent type theory with ``inductive types with conditions'' and function definitions.

\subsection{Implementations} 
\label{sub:impl}
The discussed features are implemented in the following two programming languages,
and we briefly discuss their implementations.
\begin{itemize}
\item The Arend~\cite{Arend} proof assistant.
The authors consider records with coconditions unnecessary as the only useful record with coconditions is the
$\Path$ type, so they removed the full-blown records with coconditions and built-in the $\Path$ type instead.
The termination checker \texttt{foetus}~\cite{foetus} is implemented in Arend
to ensure syntactical consistency and the decidability of type checking.
\item Old versions of the Aya~\cite{Aya} proof assistant.
Unlike Arend, it has both inductive types with conditions \textit{and} records with coconditions.
In the recent versions, it has been replaced with XTT~\cite{XTT,XTT2}, a set-level cubical type theory.
\end{itemize}

\subsection{Related work} 
\label{sub:rel-work}
Type theories based on the idea of Kan cubical sets with support for CITs have been
discussed or formalized~\cite{CHM,ICTT,CCHM,CCTT,UCTT,CHTT1,CHTT3,HIT-CCTT,CSCCTT}
and implemented~\cite{CubicalAgda,RedPRL} in many different ways by many people.
These type theories have many differences, but share a similar syntax for defining HITs.

Inductive types with conditions can specialize to the CITs defined in all of them
using the elaboration implemented in \CubicalAgda{}.
Some work (like~\cite[\S 1.5]{ICTT}) even use a condition-like syntax directly,
presenting the constructors as parameterized by intervals and attached by some reduction rules.

Rewriting type theory~\cite{RTT} where custom rewriting rules could be defined might be an even
more general case of conditions, but its logical soundness requires translation into extensional type theory~\cite[\S 6.4]{RTT},
while we are working in an intensional type theory. Apart from that, HITs are impossible to define in rewriting
type theory since path constructors could only be postulated as axioms (\cite{AgdaHoTT} uses this strategy).

The Lean~\cite{Lean} theorem prover has a semi-axiomatic (there is no way to identify
two terms with different heads in Lean without using axioms) built-in support for quotients,
but the quotient type is only a special case of HITs, not to mention inductive types with conditions.

The setoid type theory~\cite{setoidtt} and observational type theories~\cite{OTT,OTT2} can have quotients,
but these type theories do not support HITs either.

\subsection{Future work} 
\label{sub:future-work}

Type theory with conditions can be extended in many ways, including:

\begin{enumerate}
\item Conditions that are currently partial pattern matching clauses.
  We may extend it with copatterns~\cite{Copat}.
\item Similarly, we could also integrate inductive-recursive~\cite{IndRec} types or
  inductive-inductive~\cite{IndInd} types with conditions and discuss their semantics,
  or even more general inductive definitions~\cite{ToS,AKJr}.
  The prior two features are implemented in Arend already.
\item We may develop a means to perform a ``termination-check'' on the conditions in order to
  have self-referencing conditions (that refers to the constructor being attached to).
  This will allow us to define an alternative encoding of integers by the subtraction
  of two natural numbers factored by the cancellation rule.
  This will require a new proof of~\cref{lem:terminate}.
\item It would be interesting to see how conditions would interact with overlapping
  and order-independent patterns~\cite{OOP,OOP2}. For instance, we can have a constructor
  of type $\Nat→\Nat→\Dat D$ with two conditions $\zerocon,x⇒ \cons c~x$ and
  $x, \zerocon ⇒ \cons c~(\succon~x)$ (given $\cons c : \Nat → \Dat D$).
  Note that the two conditions have an overlapping case $\zerocon,\zerocon$ and one of them
  reduces to $\cons c~\zerocon$ while the other reduces to $\cons c~(\succon~\zerocon)$, 
  so pattern matching on this constructor will be checked for unusual definitional equalities.
  The overlapped case is eliminated by the case-tree translation but looks peculiar in the
  concrete syntax. Allowing overlapping but confluent patterns may improve the semantics.
\item Overlapping and order-independent patterns can make records with conditions useful.
  If there is a way to generalize over \textit{patterns} in a record, we can turn some
  simple propositional equalities definitional.
  This can dramatically simplify reasoning over structures --- for instance, if we can
  turn the identity rules in monoidal structures definitional (which is true for simple monoids
  like $\mathbb N, \mathbb Z$ but yet unsure for $\mathbb Q$, matrices, free groups, etc.),
  we can avoid a lot of rewriting that may confuse mathematicians.
\item It would be interesting to see if we can generalize coconditions to be \emph{abstract},
  say, by allowing a function input to have a cocondition, so we know partially how it computes.
  This notion can be phrased as making pattern synonyms first-class citizens in the type theory.
  This can be used to define, for instance, the type of groups where the operation computes definitionally
  on the identity element.
\end{enumerate}

\printbibliography
\end{document}